\documentclass[11pt]{article}          

\usepackage{graphicx}

\usepackage{multirow}

\usepackage[ruled]{algorithm2e}

\usepackage[noend]{algorithmic}
\usepackage{latexsym}
\usepackage{epic}
\usepackage{amsmath}
\usepackage{amssymb}
\usepackage{amsfonts}
\usepackage{theorem}
\newcommand{\rst}[1]{\ensuremath{{\mathbin\upharpoonright}%
\raise-.5ex\hbox{$#1$}}}  

\newtheorem{theorem}{Theorem}[section]
\newtheorem{definition}[theorem]{Definition}

\newtheorem{lemma}[theorem]{Lemma}
\newtheorem{corollary}[theorem]{Corollary}

\newcommand{\qed}{}
\newcommand{\bull}{\rule{.85ex}{1ex} \par \bigskip}

\newenvironment{proof}{\noindent {\bf Proof:\ }}{\hfill \bull}

\def\E{{\mathbb E}}

\def\Prob{{\mathbb P}}

\title{Feedback from Nature: \\Simple Randomised Distributed Algorithms\\for\\ Maximal Independent Set Selection and Greedy Colouring}

\date{}

\author{
Peter Jeavons$^1$, Alex Scott$^2$, Lei Xu$^1$\\[0.5cm]
{$^1$Department of Computer Science, University of Oxford, UK}\\
\texttt{\{peter.jeavons,lei.xu\}@cs.ox.ac.uk}\\[0.3cm]
{$^2$Mathematical Institute, University of Oxford, UK}\\
\texttt{scott@maths.ox.ac.uk}\\[0.5cm]
}

\newtheorem{subclaim}{Claim}{\it}{\rm}

\SetAlFnt{\small}
\SetAlCapFnt{\small}
\SetAlCapNameFnt{\small}
\SetAlCapHSkip{0pt}
\IncMargin{-\parindent}

\begin{document}

\maketitle

\begin{abstract}
We propose distributed algorithms for two well-established problems that
operate efficiently under extremely harsh conditions. 
Our algorithms achieve state-of-the-art performance in a simple and novel way.

Our algorithm for maximal independent set selection operates on a network of identical anonymous 
processors. The processor at each node has no prior information about the network. 
At each time step, each node can only broadcast a single bit to all its neighbours, 
or remain silent. 
Each node can detect whether one or more neighbours have broadcast, but cannot
tell how many of its neighbours have broadcast, or which ones.

We build on recent work of Afek et al. which was 
inspired by studying the development of a network of cells in the fruit fly~\cite{Afek2011a}.
However we incorporate for the first time another important feature of the biological system: 
varying the probability value used at each node based on local feedback from neighbouring nodes. 
Given any $n$-node network, 
our algorithm achieves the optimal expected time complexity of $O(\log n)$ rounds 
and the optimal expected message complexity of $O(1)$ single-bit messages broadcast by each node.
We also show that the previous approach, without feedback, 
cannot achieve better than $\Omega(\log^2 n)$ expected time complexity, 
whatever global scheme is used to choose the probabilities. 

Our  algorithm for distributed greedy colouring 
works under similar harsh conditions: each identical node has no prior information about the network,
can only broadcast a single message to all neighbours 
at each time step representing a desired colour, and can only
detect whether at least one neighbour has broadcast each colour value.
We show that our algorithm has an expected time complexity of $O(\Delta+\log n)$, 
where $\Delta$ is the maximum degree of the network, 
and expected message complexity of $O(1)$ messages broadcast by each node.

\end{abstract}

\section{Introduction}
\label{secIntro}
One of the most fundamental problems in distributed computing is to distributively choose a set of local
leaders in a network of connected processors so that every processor is either a leader or connected to a
leader, and no two leaders are connected to each other. This problem is known as the distributed 
{\em maximal independent set} (MIS) selection problem 
and has been considered as a challenging problem for 
three decades~\cite{Afek2011a}. 
It has many applications, especially in wireless networks~\cite{Kroeker2011b,Peleg2000d} 
and has been extensively
studied~\cite{Luby1986a,Linial1992l,Alon1986a,Kuhn2005f,Kuhn2006t,Moscibroda2005m,Metivier2011a}. 

Another fundamental problem in distributed computing that is closely related to the distributed MIS 
selection problem is the \emph{$(\Delta+1)$-colouring} problem. 
In this problem the aim is to colour the vertices of a graph which has maximum degree $\Delta$ 
using no more than $\Delta+1$ colours so that adjacent vertices are assigned different colours. 
Like the distributed MIS selection problem, the distributed $(\Delta+1)$-colouring problem 
also serves as a basic building block in many other distributed
algorithms, and has many applications for resource assignment, 
in particular for frequency assignment in radio-communication
networks~\cite{Maan2012a,Waters2005g,Ni2011c,Park1996a}. 
Because of this, it has also been extensively
studied~\cite{Barenboim2009d,Barenboim2011d,Kuhn2009w,Panconesi1996o,Schneider2010a,Johansson1999s,Panconesi2001s,Hedetniemi2003l,Hansen2004d}.

A more restricted variant of the colouring problem is called 
\emph{greedy colouring}~\cite{Grundy1939m,Gavoille2009o}, where the aim is to obtain 
a colouring with the property that no individual vertex can be recoloured with a smaller colour 
(in some fixed ordering of the colours). Computing a greedy colouring distributively is
believed to be more difficult than computing an arbitrary $(\Delta+1)$-colouring 
distributively~\cite{Gavoille2009o}, but such colourings often use a much smaller number of colours.

\subsection{Our Results}

In this paper, we first propose a randomised distributed MIS selection algorithm 
that is able to operate under very harsh conditions. Our model of distributed computing
assumes an identical anonymous processor at each node that has {\em no information about the network}.
At each time step, each node can only {\em broadcast a single bit} to all its neighbours, 
or remain silent. 
Each node can detect whether one or more neighbours have broadcast, but cannot
tell how many neighbours have broadcast, {\em or which ones}.

We prove that our algorithm is optimal in both time and bit complexity for such a model, 
by showing that it runs in expected $O(\log n)$ time, where $n$ is the number of nodes,
and the expected number of messages sent by 
each node is bounded by a constant, regardless of the network.

We then extend the approach to obtain an algorithm for the distributed greedy colouring problem. 
This algorithm also runs under very harsh conditions where the processors are anonymous and
have no information about the network. For this problem we allow each node to broadcast
only a single message to all neighbours at each time step representing a single desired colour value.
Once again nodes can only detect whether at least one neighbour has broadcast 
a colour, and cannot tell how many neighbours have broadcast, or which ones.

The algorithm we obtain is remarkably simple and computes a greedy colouring in expected $O(\Delta+\log n)$ time,
where $n$ is the number of nodes and $\Delta$ is the maximum degree of the network.
Once again the expected total number of messages sent by each processor is bounded by a constant.
As well as matching the best known time complexity for obtaining a greedy colouring, 
our algorithm is the first proposed algorithm for this problem 
where the nodes require no prior knowledge of the network
and cannot distinguish between their neighbours.

To obtain our results we introduce a new form of analysis to determine the time complexity.
Nearly all previous analytical techniques in this area have relied on a general technique, 
originally devised by Luby~\cite{Luby1986a}, which divides the computation into successive phases 
and shows that some fixed fraction of the network is expected to be eliminated in each phase, 
so that there are at most logarithmically many phases. Our algorithms do not have this property, and
hence require a more flexible form of analysis, which we describe in detail below.

\section{Preliminaries}
\label{secPre}

Given an undirected graph $G=(V, E)$, the \emph{neighbourhood} of each vertex $v \in V$ is defined to be
the set $\Gamma(v)=\{u:\{u,v\}\in E\}$ and the \emph{degree} of each vertex $v$ is defined 
to be the number $\deg_G(v)=|\Gamma(v)|$. We define the maximum degree of the graph $G$ to be the maximum
value of the degree over all vertices of $G$, which is denoted by $\Delta=\max_{v\in V}\{\deg_G(v)\}$. 
The number of vertices of $G$ is $|V|$ and will usually be denoted by $n$. 
We will say that an event on $G$ occurs \emph{with high probability} if
the probability of the event tends to 1 as $n$ tends to infinity.
We will write $\log a$ for the natural logarithm of $a$, 
and $\log_b a$ for the logarithm of $a$ to the base $b$.

\subsection{Maximal Independent Set Selection}
\begin{definition}[Maximal Independent Set]
Given an undirected graph $G = (V,E)$, an {\em independent set} in $G$ is a subset of vertices 
$U\subseteq V$, such that no two vertices in $U$ are adjacent. 
An independent set $U$ is called a {\em maximal independent set} (MIS) if no further vertex 
can be added to $U$ without violating independence.
\end{definition}

Different maximal independent sets for the same network can vary greatly in size. In contrast to the MIS selection problem, the related problem of finding a {\em maximum size} independent set (MaxIS) is notoriously hard. It is equivalent to finding a maximum clique in the complementary graph, and is therefore \textbf{NP}-hard~\cite{Karp1972r}. However, computing an arbitrary MIS (which is not necessarily of the maximum possible size) in linear time using a centralised sequential algorithm is trivial: simply scan the nodes in arbitrary order. If a node $u$ does not violate independence, add $u$ to the MIS. If $u$ violates independence, discard it. Hence the challenge is to compute such an MIS more efficiently in a distributed way with no centralised control.

\subsection{Greedy Colouring}

A proper colouring of a graph assigns a colour to each vertex such that no two adjacent vertices are assigned the same colour. The colouring is called a $k$-colouring if at most $k$ different colours are used.

\begin{definition}[Graph Colouring]
For any undirected graph $G = (V,E)$,  a {\em $k$-colouring} of $G$ is a function 
$f$ from the vertices $V$ to a set of colours $\{c_1,c_2, \ldots,c_k\}$
such that $f(u)\neq f(v)$ for every edge $\{u,v\}\in E$. 
$G$ is called {\em $k$-colourable} if and only if there exists a $k$-colouring of $G$.
\end{definition}

For many practical applications it is desirable to minimize the number of colours used.  
The smallest possible positive integer $k$ for which there exists a $k$-colouring of $G$ 
is defined to be the \emph{chromatic number} $\chi$ of $G$. 
It is known to be \textbf{NP}-hard to approximate the chromatic number $\chi$ within 
a factor of $|V|^{1-\varepsilon}$, for any $\varepsilon>0$,
even using a centralised algorithm with complete knowledge of the graph~\cite{Zuckerman2006l}.

However, a number of heuristic approaches can be used to rapidly obtain colourings
with a reasonably low number of colours on many graphs.
For example, the following greedy approach produces a colouring in linear time
using a centralised control.

\begin{definition}[Greedy Colouring]
Given an arbitrary ordering, $(v_1,v_2,\ldots,v_n)$, of the vertices of $G$,
and an arbitrary ordering on the colour values, 
a \emph{greedy colouring} algorithm considers 
each vertex from $v_1$ to $v_n$ in turn, assigning each vertex  
the smallest possible colour value that is not already assigned to any of its neighbours.
\end{definition}

\noindent
Note that a colouring obtained in this way has the property that no individual vertex can be recoloured 
using a smaller colour. A colouring with this property 
is sometimes called a \emph{Grundy colouring}~\cite{Gavoille2009o,Grundy1939m}.
Since every greedy colouring algorithm produces a Grundy colouring,
and every Grundy colouring can be obtained
using a greedy colouring algorithm (by choosing a suitable
ordering on the vertices)~\cite{Gavoille2009o},
we will refer to any Grundy colouring as a greedy colouring,
even if it is computed in some other way. 

It is easy to see that a greedy colouring uses no more than $\Delta+1$ colours,
so we have that $\chi\leq\Delta+1$ for any graph $G$. 
Brooks Theorem strengthens this observation by stating that 
$\Delta$ colours suffice for all
graphs except odd cycles and complete graphs, which require $\Delta+1$ colours.

\subsection{Distributed Computation Model}

In the widely-used Linial model~\cite{Chaudhuri1987a,Linial1986l,Linial1992l,Gavoille2009o},
a distributed network is composed of a set $V$ of processors and a set $E$ of bidirectional 
communication links (channels) between pairs of processors. 
If there is a link (channel) between two processors, these two processors are said to be {\em neighbours}. 
A distributed network with $n$ processors where each processor has no more than $\Delta$ neighbours
corresponds to an undirected graph $G=(V, E)$ with $n$ vertices and maximum degree $\Delta$. 
A network is called \emph{anonymous} if the processors cannot distinguish each other 
by unique identifiers.
Linial's distributed computation model is a synchronous system and all processors operate in a 
lockstep fashion.
We will assume that all processors wake up and start their computation at the same time step.
During each time step, all processors act in parallel and carry out the following operations
sequentially~\cite{Lynch1996d}:

\begin{enumerate}
\item Optionally send a message to each neighbouring node;

\item Receive any messages sent by neighbours;

\item Perform arbitrary local computation.

\end{enumerate}
The computation is said to be complete only when the local computations at every vertex have terminated.

The distributed computation model we use is based on this model, but 
we impose the following severe additional conditions:
\begin{enumerate}
\item Each processor is anonymous and has no local or global information about the network;
\item At each time step, each processor either keeps silent or broadcasts 
one message to all its neighbours;
\item Each processor can tell whether at least one neighbour has broadcast a message, 
but cannot tell how many of them have done so, or which ones.
\end{enumerate}

\noindent
In our MIS algorithm we restrict the communication even further, so that each message contains
only a single bit. 
In our greedy colouring algorithm we allow longer messages representing different colours.

Information about a network may be difficult to obtain, or subject to 
uncertainty and change,
so it is desirable for some applications to find algorithms that can complete their task without 
using such information~\cite{Halpern1990k,Prakash1997a}.
Moreover, using a small number of messages, each containing a single bit (or a small number of bits),
allows an implementation to use less communication resources and less energy, 
and this may be crucial in some applications~\cite{Lenzen2012d}.
Because of the restrictions we impose, 
our algorithms can be implemented using very simple communication
mechanisms such as radio waves, optical signals, or even chemical signals, as in biological
intercellular signalling~\cite{Bray2006n,Collier1996p}.

\section{Related Results}

\subsection{Distributed MIS selection}

The study of distributed MIS selection can be traced back to the 1980s. 
It was shown early on that the MIS selection problem is in the complexity class \textbf{NC}~\cite{Karp1985a}, and hence likely to be a good candidate for a parallel
or distributed approach. 

We review the current state-of-the-art here, focusing on the size of the messages (in bits)
and the information about the network that is used at each node (see Table~\ref{TalgorithmsMIS}). 
In many cases it is possible to use estimates for the required graph parameters, and to 
iteratively refine these, at the cost of a more sophisticated algorithm and additional 
communication rounds, but we describe only the simplest versions of the algorithms,
as originally presented.

\begin{table*}[th]
\renewcommand{\arraystretch}{1.3}
\caption{Distributed MIS selection algorithms on graphs with $n$ nodes and maximal degree $\Delta$\label{TalgorithmsMIS}}
\vspace{2mm}
\scalebox{0.9}{
\begin{tabular}{|l|l|l|l|l|}
      \hline
         Type                         &                    Time Steps                                   &             Message size (bits)                      &                       \parbox{6cm}{Information about the graph and neighbourhood used at each node}                                     &   Reference \\\hline\hline
\multirow{3}{*}{Det.}     & \multirow{2}{*}{$O(\Delta+\log^* n)$}           & \multirow{3}{*}{$\Omega(\log n)$}       &\multirow{3}{*}{\parbox{55mm}{Unique IDs, size and maximum degree of the graph, and distinguishable channels}}  & \cite{Barenboim2009d}\\\cline{5-5}           
                                         &                                                                           &                                                                  &                                                                                                                    &\cite{Kuhn2009w}\\\cline{2-2}\cline{5-5}
                                         & \multirow{1}{*}{$O(2^{O(\sqrt{\log n})})$}    &                                                                &                                                                                                                     & \cite{Panconesi1996o}\\ \cline{1-5}
\multirow{9}{*}{Rand.} & \multirow{4}{*}{$O(\log^2 n)$}                      &\multirow{1}{*}{$\Omega(\log \Delta)$} &  Maximum degree in 2-neighbourhood                                                       &\cite{Peleg2000d}\\ \cline{3-5}                                                                   
                                         &                                                                            &  \multirow{1}{*}{3}                                 & \multirow{1}{*}{None}                                                                             &\cite{Emek2013s}\\ \cline{3-5}
                                         &                                                                            & 1                                                               & Size of the graph                                                                                      &\cite{Afek2011a}\\ \cline{3-5}
                                         &                                                                            &  1                                                               &  None                                                                                                       & \cite{Afek2011b}\\ \cline{2-5}
                                         & \multirow{1}{*}{$O(\log \Delta\sqrt{\log n})$} & \multirow{1}{*}{$\Omega(\log n)$}        &  \parbox{6cm}{Size, maximum degree of the graph, and distinguishable channels}                       &\cite{Barenboim2012t}\\ \cline{2-5}
                                         & \multirow{4}{*}{$O(\log n)$}                            &  {$\Omega(\log n)$}                                &  Size of the graph                                                                                        &\cite{Lynch1996d}\\ \cline{3-5}
                                         &                                                                            &   $\Omega(\log \Delta)$                          & Degrees of neighbours                                                                                &\cite{Wattenhofer2007l}\\ \cline{3-5}
                                         &                                                                           &  1                                                              & Distinguishable channels                                                                                           &\cite{Metivier2011a}\\ \cline{3-5}
                                         &                                                                          & 1                                                                & \textbf{None}                                                                                              &\textbf{This paper}\\ \cline{1-5}
\end{tabular}}
\end{table*}

A lower bound of $\Omega(\log^* n)$ time for distributed MIS selection on graphs with 
$\Delta \ge 2$ is given in~\cite{Linial1987d}. 
The most well-known lower bound for distributed MIS selection on general graphs, 
$\Omega(\sqrt{\log n/(\log\log n)})$, is given in~\cite{Kuhn2004w}. 
This was improved to $\Omega{(\min\{\log \Delta,\sqrt{\log n}\})}$
in~\cite{Kuhn2010} (see also~\cite{Barenboim2012t}).
All of these lower bounds have been shown to apply to both deterministic and randomised algorithms.
It was observed in~\cite{Metivier2011a} that if only one-bit messages are allowed 
to be sent along each edge in any time step, then every distributed algorithm 
to select an MIS in a ring of size $n$ requires at least $\Omega(\log n)$ time steps 
with high probability.

For deterministic distributed MIS selection on general graphs, the fastest known algorithms run in
$O(\Delta+\log^* n)$ time~\cite{Kuhn2009w,Barenboim2009d} 
or $O(2^{O(\sqrt{\log n})})$ time~\cite{Panconesi1996o}. 
These deterministic MIS algorithms rely on very sophisticated multi-phase techniques, 
use a considerable amount of global information about the graph at each node, 
including unique node IDs,
and allow complex messages to be sent on specific channels between nodes.
Note that any deterministic algorithm requires some information at each node 
(such as a unique node ID) in order to break the symmetry~\cite{Itai1990s,Peleg2000d}. 

Using randomisation to break symmetry between nodes allows for simpler
algorithms, often requiring a smaller number of time steps.
A simple parallel randomised algorithm for distributed MIS selection in the PRAM model of computation
was presented in 1986 by Luby~\cite{Luby1986a} and independently by Alon et al.~\cite{Alon1986a}.

This algorithm has been adapted to the message-passing model of distributed computation in several
slightly different ways. In the version presented by Lynch~\cite{Lynch1996d} each processor
is assumed to know the total size, $n$, of the graph, and chooses a random integer in the range 
1 to $n^4$ at each time step. These integers are then broadcast as messages to all neighbouring nodes,
so the messages sent between processors contain $\Omega(\log n)$ bits.
Using these messages the nodes are able to compute an MIS by selecting the nodes that 
choose the largest random values in their neighbourhood, removing those nodes and their neighbours, 
and iterating this process. Using the analysis from~\cite{Luby1986a}, this process 
is shown to terminate in $O(\log n)$ time on average and with high
probability.

In the version presented by Wattenhofer~\cite{Wattenhofer2007l} the nodes choose a probability value based on their
degree in the graph, and use this value, together with the degree values of their neighbours to 
decide whether to join the MIS at each time step. In this variant the nodes
exchange messages to
determine the current degrees of their neighbours at each time step, and hence the 
messages sent between processors contain $\Omega(\log \Delta)$ bits.
Once again, using the analysis from~\cite{Luby1986a}, this process 
is shown to terminate in $O(\log n)$ time on average and with high probability.

In the version presented by Peleg~\cite{Peleg2000d} the probability value
at each node is chosen based on the maximum degree of the nodes 
at distance 1 or 2 away from it in the graph,
and this value is then used to decide whether to join the MIS at each time step. 
Peleg shows with a simpler analysis that this algorithm halts in $O(\log^2 n)$ time
on average and with high probability.
Once again the nodes
exchange messages to determine the current degrees of their neighbours 
at each time step, and hence the 
messages sent between processors contain $\Omega(\log \Delta)$ bits.

These distributed randomised algorithms, 
all based on a similar approach and generally known as Luby's algorithm, 
remained the state-of-the-art for many years, but there have recently been some new developments.
 
A new randomised MIS algorithm with time complexity $O(\log \Delta\sqrt{\log n})$
was proposed in~\cite{Barenboim2012t}. 
This algorithm improves on the $O(\log n)$ algorithms when $\log \Delta < \sqrt{\log n}$. 
On the other hand, this algorithm assumes that each processor knows 
the size and maximal degree of the graph and can distinguish 
between channels so that it can send different messages along different edges.
Since it relies on exchanging information about specific nodes,
using node identities, the messages exchanged in this algorithm contain $\Omega(\log n)$ bits.

Algorithms for MIS selection on special graphs such as sparse graphs and growth-bounded graphs 
have also been studied~\cite{Goldberg1988p,Barenboim2010s,Schneider2008a}. 

\subsection{MIS Selection with Limited Communication}

There has recently been considerable interest in finding 
efficient distributed MIS selection 
algorithms that can work in more restricted computational models, 
such as wireless network models~\cite{Moscibroda2005m,Emek2013s,Cornejo2010d,Afek2011a,Afek2011b,Metivier2011a,Xu2013s}. 

For example, the approach proposed in~\cite{Metivier2011a}
splits the randomly generated values at each node into single bits, and 
communicates them one by one. When these bits are broadcast to all neighbours, 
this approach achieves a time complexity of $O(\log^2 n)$. By distinguishing
between different neighbours, and having separate, overlapping, exchanges of messages with 
each neighbour, the overall time complexity is brought down to $O(\log n)$ time 
on average and with high probability. 
This is shown to be the optimal time complexity that can be achieved with one-bit messages~\cite{Metivier2011a}.
However, to achieve this optimal performance requires that each vertex can 
distinguish between its neighbours by locally known channel names, 
so that different messages can be sent along different edges at the same time step.

A more radical approach is the novel distributed MIS selection algorithm inspired by the neurological development of the fruit fly which is given in~\cite{Afek2011a,Afek2011b}.

During development, certain cells in the pre-neural clusters of the fruit fly specialise to become 
sensory organ precursor (SOP) cells, which later develop into cells attached to 
small bristles (microchaetes) on the fly that are used to sense the environment. 
During the first stage of this developmental process each cell either becomes an SOP or a neighbour 
of an SOP, and no two SOPs are neighbours. These observed conditions are identical to the formal
requirements in the maximal independent set selection problem.

However, Afek \textit{et al.} pointed out that the method used by the fly to select the SOPs 
appears to be rather different from the standard algorithms for choosing an MIS described above.
The cells of the fly appear to solve the problem using only simple local interactions between certain
membrane-bound proteins, notably the proteins Notch and Delta~\cite{Bray2006n,Collier1996p}. 
Moreover, they require very little knowledge about connectivity.
Based on their study of this developmental process, Afek \textit{et al.} proposed an
algorithm that works in a 
distributed model where each node can only broadcast to all its neighbours or remain silent.
Moreover, each node can only detect whether at least one neighbour has broadcast a signal.
This model of communication is sometimes referred to as a ``beeping'' model 
with collision detection~\cite{Afek2011b}.

In their proposed algorithm, each node broadcasts at each time step with a certain probability,
which changes over time, and then checks whether any of its neighbours has broadcast at the same time. 
As originally presented \cite{Afek2011a}, the algorithm
uses a sequence of gradually increasing global probability values
calculated from the total number of nodes of the graph $n$ and its maximum degree $\Delta$.
The algorithm was further refined by Afek \textit{et al.} in a later paper \cite{Afek2011b}.
In the later version the probability values are chosen according to a fixed pattern, so that the individual nodes require no information at all about the graph.
However, in both versions the expected number of time steps required 
was shown to be $O(\log^2 n)$ (see Table~\ref{TalgorithmsMIS}). 

Another approach to distributed computing  
with very restricted communication and processing capabilities is
the networked finite state machine model introduced in~\cite{Emek2013s}.
This only allows a fixed finite number of distinct messages, and very limited computation 
at each node, based on the notion of a randomised finite state machine, with no information 
about the network. 
It is shown in~\cite{Emek2013s} that a MIS can be computed in this very restricted model 
in $O(\log^2 n)$ time, using only 7 states and 7 corresponding messages.

\subsection{Distributed Colouring}

The problem of $(\Delta + 1)$-colouring is closely related to MIS selection
\cite{Kuhn2006o}. 
Hence it can be shown that 
the lower bounds for distributed MIS selection mentioned above also apply for distributed 
$(\Delta+1)$-colouring. 
Similarly, some state-of-the-art distributed $(\Delta+1)$-colouring algorithms 
are closely related to the algorithms for distributed MIS selection described 
earlier (see Table~\ref{TalgorithmsGC}).

\begin{table*}[t]
\renewcommand{\arraystretch}{2}
\caption{Distributed $(\Delta+1)$-colouring algorithms on general graphs with $n$ nodes and maximal degree $\Delta$ \label{TalgorithmsGC}}
\vspace{2mm}
\small
\scalebox{0.85}{
   \begin{tabular}{|l|l|l|l|l|l|}
 \hline
      Greediness & Type & Time Steps  & \parbox{12mm}{Message size (bits)} & Information used at each node& Reference \\\hline\hline
 \multirow{8}{*}{Non-greedy} &\multirow{3}{*}{Det.} & \multirow{2}*{$O(\Delta+\log^*n)$}               &{\multirow{3}{*}{$\Omega(\log n)$}}& \multirow{3}{*}{\parbox{60mm}{Unique IDs, size and maximum degree of the graph, and distinguishable channels}} & \cite{Barenboim2009d} \\ \cline{6-6}
                                                    &                                     &                                                                           &                                                             &                                                                                                                  & \cite{Kuhn2009w}\\ \cline{3-3}\cline{6-6}
                                                    &                                     & \multirow{1}{*}{$O(2^{O(\sqrt{\log n})})$}   &                                              &                                                                                                                    & \cite{Panconesi1996o}\\ \cline{2-6}

            & \multirow{5}{*}{Rand.}  & \multirow{1}{*}{$O(\log \Delta+\sqrt{\log n})$}                    & \multirow{1}{*}{$\Omega(\log \Delta)$} & \parbox{60mm}{Upper bound on the size of the graph and distinguishable channels} & \cite{Schneider2010a}\\ \cline{3-6}
            &                                        & \multirow{1}{*}{$O(\log \Delta+2^{O(\sqrt{\log \log n})})$} & \multirow{1}{*}{$\Omega(\log n)$} & {\parbox{60mm}{Unique IDs, size and maximum degree of the graph, and distinguishable channels}} & \cite{Barenboim2012t}\\ \cline{3-6}
            &                                        &  \multirow{3}{*}{$O(\log n)$}                                                  &  $\Omega(\log n)$                             & Size of the graph                                                                                       &\cite{Lynch1996d}\\\cline{4-6}
             &                                       &                                                                                                  & \multirow{2}{*}{$\Omega(\log \Delta)$} & \multirow{1}{*}{Degrees of neighbours}                                            & \cite{Wattenhofer2007l}\\\cline{5-6}
            &                                        &                                                    &                                                                                                                 &  Maximum degree of the graph                                                           & \cite{Johansson1999s}\\ \hline \hline

 \multirow{5}{*}{Greedy}        & \multirow{1}{*}{Det.} & \multirow{1}{*}{$O(\Delta^2+\log^*n)$}  &  \multirow{1}{*}{$\Omega(\log n)$} &  {\parbox{60mm}{Unique IDs, maximum degree of the graph, own degree and distinguishable channels}} & \cite{Panconesi2001s}\\ \cline{2-6}
                                                 &\multirow{4}{*}{Rand.}& \multirow{1}{*}{$O(\Delta^2\log n)$}       & \multirow{1}{*}{$\Omega(\log \Delta)$} &  Own degree and degrees of neighbours                                            & \cite{Hansen2004d} \\ \cline{3-6}                                                        
                                                 &                                      & \multirow{3}{*}{$O(\Delta+\log n)$}           & \multirow{1}{*}{$\Omega(\log \Delta)$} & Maximum degree of the graph                                                            & \cite{Gavoille2009o} \\  \cline{4-6}
                                                 
            &  &   &      \multirow{2}{*}{\parbox{20mm}{$O(\log \mu)$\\ (where $\mu =$ max colour used)}}            &  \multirow{1}{*}{Distinguishable channels}                                                     & \cite{Metivier2010a}\\ \cline{5-6}
                                                &                                         &                                                                      &                                                                & \textbf{None}                                                                                    &\textbf{This paper}\\ \hline   
    \end{tabular}}
\end{table*}

A randomised distributed $(\Delta+1)$-colouring algorithm requiring $O(\log \Delta+\sqrt{\log n})$ time
is proposed in~\cite{Schneider2010a}. In this algorithm 
the messages represent randomised preference levels for each of the possible colours. 
This algorithm needs to know an upper bound of the size of the graph
and requires each processor to be able to send different messages along different channels. 
Since the messages exchanged represent colours, the message size is at least $\Omega(\log \Delta)$.

A randomised algorithm with expected time complexity of 
$O(\log \Delta+2^{O(\sqrt{\log \log n})})$ 
is given in~\cite{Barenboim2012t}. However, this algorithm relies on a deterministic algorithm to 
complete a partial colouring, and hence requires unique node IDs, and messages with $\Omega(\log n)$ bits.

Johansson proposed and analysed a simple randomised distributed 
$(\Delta + 1)$-colouring algorithm requiring $O(\log n)$ time~\cite{Johansson1999s}. 
The algorithm of Johansson requires that each vertex knows the maximum degree of the graph.
Each message corresponds to a potential colour choice, so the messages in this algorithm 
contain $\Omega(\log \Delta)$ bits. 

These algorithms do not attempt to obtain {\em greedy} colourings, and hence tend to use the 
maximum number, $\Delta+1$, of colours.
For many classes of graphs, a greedy colouring will often use considerably fewer colours,
but computing a greedy colouring with a distributed algorithm is a more challenging problem.
In fact, the problem of computing a 
greedy colouring for a given ordering of the vertices is known to be 
\textbf{P}-complete~\cite{Greenlaw1995l,Gebremedhin2003g}. 

Panconesi and Rizzi proposed a deterministic algorithm for graph colouring 
that attempts to use a small number of colours~\cite{Panconesi2001s}. 
This algorithm was not originally designed to construct a greedy colouring. However, 
it can be easily modified to become a distributed greedy colouring algorithm 
by always choosing the first available colour when assigning a colour. 
In view of this it is described in~\cite{Gavoille2009o} as the first distributed approach to greedy colouring. The number of time steps taken by this algorithm is $O(\Delta^2+\log^*n)$~\cite{Panconesi2001s}.
It is quite a sophisticated algorithm which relies on a preprocessing phase to produce a forest
decomposition of the graph. It assumes that each vertex has a unique identifier, 
and it also requires that each vertex knows its own degree and the maximum degree of the whole graph. 
This algorithm also requires the ability to send different messages to different neighbours simultaneously.
Because identifiers are exchanged the message size of the algorithm is $\Omega(\log n)$.

Hansen et al. proposed a randomised distributed algorithm for graph colouring in~\cite{Hansen2004d}. 
Even though the algorithm is not explicitly described in the original paper as a greedy colouring
algorithm, it is pointed out in~\cite{Gavoille2009o} that the colourings it produces
are actually greedy colourings. The expected number of time steps taken by this algorithm to produce a
colouring is $O(\Delta^2\log n)$~\cite{Hansen2004d}.
However, this algorithm assumes that each vertex knows its degree in the graph, 
and the messages exchanged include these numerical degree values as well as the colour values. 
Hence the size of each message sent is at least $\Omega(\log \Delta)$ bits. 

Gavoille et al. give a detailed theoretical study of distributed greedy colouring~\cite{Gavoille2009o}.
They establish a lower bound for this problem of $\Omega(\log n/\log\log n)$ time steps.
Moreover, they note that an arbitrary $k$-colouring can be converted to a greedy colouring by a 
simple distributed algorithm in $O(k)$ time steps. However, the conversion algorithm 
in~\cite{Gavoille2009o} requires 
each node to know the value of $k$. Combining this approach with the most efficient
$(\Delta + 1)$-colouring algorithms described earlier gives a two-stage algorithm with
an overall expected time complexity of $O(\Delta + \log n)$.

M{\'e}tivier et al. proposed a simple randomised distributed $(\Delta + 1)$-colouring algorithm requiring $O(\Delta+\log n)$ time~\cite{Metivier2010a}. The algorithm proposed by M{\'e}tivier et al., does not assume any global knowledge of the network, but requires each processor to know from which channel it receives each message. 
The algorithm consists of two stages: it first uses randomisation to break the symmetry and
obtains a colouring in $O(\log n)$ time with an unbounded number of colours; 
it then reduces the number of colours used to at most $\Delta+1$ in $O(\Delta+\log n)$ time.
The colours in this second stage are chosen to be the smallest available, so the resulting colouring 
is a greedy colouring (although this is not made explicit).
Each message exchanged in the first stage contains only one bit, 
but each message in the second stage represents a final colour choice.
Since the total number of colours used may be much lower than $\Delta$ in some classes of graphs,
we give an upper bound on the message size of $(\log \mu)$ bits for this algorithm,
where $\mu$ is the maximum colour number used.

\section{Algorithm for MIS Selection}
\label{secMIS}

\begin{table*}[th]
\caption{The algorithm for distributed MIS selection at each node
\label{alg:MISlocal}}{
\centering
\framebox{
\begin{minipage}[c]{0.85\linewidth}
\begin{tabbing}
{\bf Global constants:} \= $p_0$ : lower bound on initial probability value; \\
         \> $f_1,f_2$ : lower and upper bounds on change factor for probability value.\\[0.1cm]
{\bf Local variables:} \> $p$ : local probability value, initialised to some value in $[p_0,1]$;\\
         \> $f$ : change factor for probability value, chosen arbitrarily in $[f_1,f_2]$;\\
         \> {\sc Trying} : Boolean flag, initialised to {\sc False}.
\end{tabbing}
\algsetup{
linenosize=\tiny,
linenodelimiter=.,
indent = 3em
}
\renewcommand{\algorithmiccomment}[1]{\hfill {\sc \small #1}}
\begin{algorithmic}[1]
\WHILE{active, in each round}
\vspace*{0.2cm}
\STATE *FIRST EXCHANGE*
\STATE With probability $p$, 
set {\sc Trying} $\leftarrow$ {\sc True} and {\bf send} signal to all neighbours;
\STATE {\bf Receive} any signals sent by neighbours;
\STATE Set $f$ to some value in the interval $[f_1,f_2]$;
\IF{any signal was received}
   \STATE {\sc Trying} $\leftarrow$ {\sc False} and $p\leftarrow p/f$ (decrease $p$)  
\ELSE
   \STATE $p\leftarrow \min\{fp,1\}$ (increase $p$)
\ENDIF

\vspace*{0.2cm}
\STATE *SECOND EXCHANGE*
\IF{{\sc Trying}}
    \STATE {\bf Send} signal to all neighbours;
    \STATE Join the MIS and terminate (become inactive).
\ENDIF
\STATE {\bf Receive} any signals sent by neighbours;
\IF{any signal was received}
  \STATE Terminate (become inactive)
\ENDIF
\ENDWHILE
\vspace*{0.2cm}
\end{algorithmic}
\end{minipage}
}
}
\end{table*}

The distributed algorithm for MIS selection proposed by Afek \textit{et al.} 
is remarkably simple~\cite{Afek2011a}.
At each step, each node may choose, with a certain probability $p$ (that varies over time),
to signal to all its neighbours that it wishes to join the independent set.
If a node chooses to issue this signal, and none of its neighbours choose to do so in the same time step,
then it successfully joins the independent set, and becomes inactive, 
along with all its immediate neighbours.
However if any of these neighbouring nodes issue the same signal at the same time step, 
then the node does not succeed in joining the independent set
at that step. This process is repeated until all nodes become inactive.

Our new algorithm uses a similar basic scheme, but with a different approach to
the way that the probability value $p$
varies over time (see Table~\ref{alg:MISlocal}).
Inspired by the positive feedback mechanisms that control cellular processes~\cite{Bray2006n,Collier1996p},
we give each node an independently updated probability value. 
These probabilities are initialised to arbitrary values 
(above some fixed threshold value, $p_0 > 0$).
They are decreased whenever one or more neighbouring nodes signal that they 
wish to join the independent set, and are increased 
whenever no neighbouring node issues such a signal.
We allow each increase or decrease to be by some arbitrary factor $f$, which may vary at each step,
but is bounded by the global parameters $f_1$ and $f_2$ (with $1 < f_1 \leq f_2$).

Our main result below shows that varying the probabilities in this way, 
using a simple local feedback mechanism,
gives an algorithm whose expected time to compute a maximal independent set is $O(\log n)$ (see Corollary~\ref{logexpected}, below).
We also show that the expected number of signals sent by each node is 
bounded by a constant (see Theorem~\ref{thm:expectedbeeps}, below).

Note that the algorithm in Table~\ref{alg:MISlocal} consists of two successive message exchanges.
We shall refer to each such pair of message exchanges as a {\em round} of the algorithm.
Hence each round occupies two consecutive time steps.

To investigate the performance of our new algorithm in practice we 
constructed an implementation with the probability $p$
at each node varying as follows: $p$ is initially set to ${1}/{2}$. 
In any round where a signal is received from at least one neighbouring cell the value of $p$ is halved. 
In all other rounds it is doubled (up to a maximum of $1$).
We then compared this algorithm with the algorithm in~\cite{Afek2011b} by running both of them on random networks with different numbers of nodes, 
where each edge is present with probability ${1}/{2}$~\cite{Scott2013f}.

We found that the mean number of rounds required in our experiments to complete our algorithm 
and choose a maximal independent set in these networks was approximately $2.5 \log_2 n$,
for all values of $n$ between 20 and 200. However, the mean number of rounds required by the algorithm in~\cite{Afek2011b} to select a maximal independent set was close to the exact value of $\log_2^2 n$.
Afek et al. do not discuss the expected number of signals broadcast at each node in their algorithm. 
We found that the mean number of signals sent by each node in our algorithm was less than 2, 
regardless of the size of the network. However, our experiments indicated that the mean number 
of signals sent by each node when running the algorithm described in~\cite{Afek2011b} 
increased with the size of the network.

Before we analyse the performance of this algorithm 
we first demonstrate in Section~\ref{sec:MISgloballower}
that the use of a feedback mechanism to adjust the probability values,
as described in lines 5-9 of Table~\ref{alg:MISlocal}, is crucial
to achieving the efficiency.

\subsection{Lower Bound for Globally Chosen Probability Values}
\label{sec:MISgloballower}

In this section we consider a class of algorithms similar to the one described in~\cite{Afek2011a}
where each node runs through the same fixed preset sequence of probability values,
and does not adjust these to take into account the behaviour of other nodes.
In other words, we consider a simplified version of the algorithm described in 
Table~\ref{alg:MISlocal}, where the probability values at all nodes are initialised to 
the same value $p_0$,
and the probability updates described in lines 5-9 are replaced
by a simple update rule that changes $p$ to the next value in some fixed sequence $p_1,p_2,\dots$.
We refer to this modified algorithm as 
\emph{MIS selection with global probability values}.

Our first result constructs an explicit family of graphs with $O(n)$ vertices, for which
{\em any} such algorithm takes at least $\Omega(\log^2 n)$ rounds,
no matter what sequence of probability values is used.
(Note that we generally omit floors and ceilings for clarity, and the graphs we construct
in this result have $O(n)$ vertices rather than exactly $n$ vertices, to simplify their
description.)
\begin{theorem}
There is a constant $\kappa>0$ such that the following holds.
Let $G$ be the graph consisting of $n^{1/3}$ disjoint copies of the complete graph $K_d$, 
for each $d=1,\dots,n^{1/3}$.
Then with high probability,
any MIS selection algorithm with global probability values 
running on $G$ does not terminate within $\kappa \log^2 n$ rounds.
\end{theorem}
\begin{proof}
Let $p_0,p_1,p_2,\ldots,$ be the sequence of probability values used by the algorithm.  
Fix $d$, and consider a copy $K$ of $K_d$.  The probability that some vertex of $K$ is added to the
independent set at the $i$th round is the probability that exactly one vertex of $K$ beeps, 
and so equals
\begin{equation}
\label{eq:problimit}
dp_i(1-p_i)^{d-1} \le dp_i\exp(-(d-1)p_i).
\end{equation}
Note that the function $xe^{-x}$ is bounded on $[0,\infty)$, and has maximum $1/e$ (at $x=1$).  
So for $d>2$,
$$
dp_i\exp(-(d-1)p_i) = \frac{d}{d-1} \cdot  (d-1)p_i\exp(-(d-1)p_i)\le\frac{3}{2e}.
$$
Also, for $x\in[0,3/2e]$, we have $1-x\ge\exp(-2x)$.
So, by inequality~\eqref{eq:problimit}, the probability that all the vertices of $K$ are still active
after $T$ rounds is at least
\begin{align*}
\prod_{i=1}^T \big(1-dp_i e^{-(d-1)p_i}\big)
&\ge \prod_{i=1}^T \exp(-2dp_i e^{-(d-1)p_i}) \\
&= \exp(-\sum_{i=1}^T 2dp_i e^{-(d-1)p_i}) \\
&\ge \exp(-\sum_{i=1}^T 6dp_i e^{-dp_i}).
\end{align*}
The last inequality follows from the fact that $e^{p_i} \leq e \leq 3$.

Hence if $\sum_{i=1}^T 6dp_i e^{-dp_i}<\frac 14\log n$ then the nodes of $K$ remain active 
with probability at least $n^{-1/4}$.
In that case the probability that the nodes in all the copies of $K_d$ become inactive 
in $T$ rounds is at most
$$(1-n^{-1/4})^{n^{1/3}} \le \exp (-n^{1/12}),$$
and so with high probability the algorithm fails to terminate in $T$ rounds.

It follows that we need only consider the case when
$$\sum_{i=1}^T 6dp_i e^{-dp_i} > \frac 14\log n$$ 
for every choice of $d\ge3$.  We will show that this implies $T=\Omega(\log^2 n)$.

Let us choose $d$ at random.  We define a probability distribution for $d$ by
$$\Prob[d=j]= \frac{c}{j\log n},$$
for $j=3,\ldots,n^{1/3}$ (where $c$ is a normalizing constant: note that
$c=\Theta(1)$, as $\sum_{i=1}^{n^{1/3}}1/j=\Theta(\log n)$).
Then, for any $p\in[0,1]$,
$$
\E[dpe^{-dp}]
= \sum_{j=3}^{n^{1/3}} \frac{c}{j\log n} jpe^{-jp}
\le \frac{c}{\log n} \sum_{j=0}^\infty pe^{-jp}.
$$
But $\sum_{j=0}^\infty pe^{-jp} = p/(1-e^{-p})<2$, as $p\in[0,1]$; 
so we have $\E[dpe^{-dp}]< 2c/\log n$.  
By linearity of expectation,
choosing a random $d$, we have 
$$\E\left[\sum_{i=1}^T 6dp_i e^{-dp_i}\right] < 12cT/\log n.$$  
Hence there is some value of $d$ for which
$$\sum_{i=1}^T 6dp_i e^{-dp_i}< 12cT/\log n.$$ 
By the argument above, this quantity must be at least
$\frac 14\log n$, and so we must have $T=\Omega(\log^2 n)$.
\hfill~
\qed
\end{proof}

\subsection{Time Complexity with Locally Chosen Probability Values and Feedback}
\label{sec:localtime}

In this section we analyse the running time of our new algorithm for distributed MIS selection
described in Table~\ref{alg:MISlocal}, where the probability values at each node 
are locally varied in each round based on feedback from neighbouring nodes.

It follows from the analysis of~\cite{Afek2011a} that if this algorithm terminates 
(i.e., all nodes become inactive) then it correctly identifies an MIS.
The only question is the number of rounds required.

Note that, unlike Luby's algorithm \cite{Alon1986a,Luby1986a}, it is not true that in 
every round we can expect at least some constant fraction of the edges to be incident 
to nodes that become inactive in that round. 
For example, in a complete graph nodes will only become inactive when exactly one node signals. 
If all nodes are initialised with the same probability value and with the same (fixed) 
increase and decrease factor $f$, then all nodes will always possess the same probability value $p_t$.
Whenever more than one node signals, all nodes will decrease their probabilities by $f$; 
if no node signals, all nodes will increase their probabilities by $f$. 
The probability of exactly one node signalling is thus $np_t(1-p_t)^{n-1}$ at each round $t$.
Hence, for complete graphs, 
with high probability all nodes will remain active for any fixed constant number of rounds.
It follows that we must carry out a more detailed analysis over a sequence of rounds
whose length increases with $n$.

\begin{theorem}\label{logrounds}
For any fixed values of $p_0 > 0$, and $1 < f_1 \leq f_2$, 
there is a constant $K_0$ such that the following holds:
For any graph $G$ with $n$ vertices, and any $k\ge1$,
the algorithm defined in Table~\ref{alg:MISlocal} terminates in at most $K_0 (k+1) \log n$  rounds,
with probability at least $1-O(1/n^k)$.
\end{theorem}

Before beginning the proof of Theorem \ref{logrounds}, it will be useful to 
define some notation and record a few simple facts.
We will frequently use the well-known inequality
\begin{equation}\label{exp}
(1-\delta) \le \exp(-\delta).
\end{equation}

We will also use the following inequality, which holds for any $\lambda > 0$ 
and any $\delta \in [0,1-e^{-\lambda}]$ (it holds with equality at the ends of this interval, and 
so holds at all points in between, by convexity).
\begin{equation}\label{exp2}
(1-\delta) \ge \exp(-\delta \lambda/(1-e^{-\lambda})).
\end{equation}


Finally, we will also need the following Chernoff-type inequality: 
if $X$ is a sum of Bernoulli random variables,
with expected value $\E X=m$, then for every $\delta>0$,
$$\Prob[X>m+\delta]\le\exp(-\delta^2/(2m+2\delta/3)).$$
In particular,
\begin{equation}\label{bigup}
\Prob[X>2m]\le\exp(-m/3).
\end{equation}

We refer to sending a signal in the first exchange of the algorithm in Table~\ref{alg:MISlocal}
as ``beeping", and receiving such a signal from a neighbour as ``hearing a beep". 

For any vertex $v$, we define $\mu_t(v)$, which we call the ``weight'' of $v$, 
to be the probability that $v$ beeps in round $t$. 
(By convention, we set $\mu_t(v)=0$ if $v$ is inactive at time $t$; this
simplifies notation, while allowing us to ignore the contribution of inactive vertices.)
For any $W \subseteq V$ we write $\mu_t(W)$ for $\sum_{v \in W} \mu_t(v)$.
Note that $\mu_t$ is a random measure on $V$, 
as it depends on the beeps of other vertices during the first $t-1$ rounds.

Recall that the set of vertices adjacent to a given vertex $v$ is called the set of
{\em neighbours} of $v$, and denoted by $\Gamma(v)$. 

\begin{definition}\label{def:lambdalightMIS}
For any $\lambda > 0$, a vertex $v$ will be called {\em $\lambda$-light} in round $t$ if
$\mu_t(\Gamma(v))\le\lambda$ and every neighbour of $v$ has weight at most $1-\exp(-\lambda)$;
otherwise, vertex $v$ is called {\em $\lambda$-heavy}.
\end{definition}

For any vertex that is {\em $\lambda$-light}, 
the weight of each of its neighbours individually is bounded by $1-\exp(-\lambda)$ 
and the sum of all its neighbours' weights is not too large 
(and so the vertex is not too likely to hear a beep at time $t$). 
Note that a fixed vertex may move back and forth between being $\lambda$-heavy and 
$\lambda$-light over time.

Our first result establishes a lower bound on the probability that at least one vertex in a set
of $\lambda$-light vertices will be added to the independent set in the current round.

\begin{lemma}\label{light}
Let $W$ be a set of vertices that are $\lambda$-light at round $t$. 
The probability that at least one vertex in $W$ is added to the independent set in round $t$ 
is at least $e^{-\phi\lambda}(1-e^{-\mu_t(W)})$ where $\phi = \lambda/(1-\exp(-\lambda))$.
\end{lemma}
\begin{proof}
Let us order the vertices of $W$ as $w_1,\ldots,w_m$, where $m = |W|$.  
The probability that some vertex of $W$ is added to the independent set in round $t$ 
is at least the probability that the smallest vertex of $W$ that beeps in round $t$ 
is added to the independent set.
For $i=1,\dots,m$, define events $E_i$ and $F_i$ by
\[
E_i=(\mbox{$w_i$ beeps; $w_1,\dots,w_{i-1}$ do not beep})
\]
\[
F_i=(\mbox{no neighbour of $w_i$ beeps}).
\]
The events $E_i\cap F_i$ are pairwise disjoint, so using the definition of conditional probability,
we have that the probability that the smallest of $W$ that beeps 
is added to the independent set is
$$
\Prob\left[\bigcup_{i=1}^m(E_i\cap F_i)\right]
=\sum_{i=1}^m \Prob[E_i\cap F_i]
=\sum_{i=1}^m  \Prob[E_i]\Prob[F_i |E_i].
$$
It is easily seen that $\Prob[F_i|E_i]\ge \Prob[F_i]$ since $\Prob[F_i|E_i]$ is conditioned on the
event that $w_i$ beeps and $w_1,\dots,w_{i-1}$ do not. 
Hence we have 
$$\Prob[F_i |E_i] \ge \Prob[F_i] = \prod_{v\in \Gamma(w_i)}(1-\mu_t(v))$$
Since $w_i$ is $\lambda$-light, we may apply Inequality~\eqref{exp2}, to conclude that 
\begin{align*}
\prod_{v\in \Gamma(w_i)}(1-\mu_t(v)) 
& \ge \prod_{v\in \Gamma(w_i)}\exp(-\phi\mu_t(v))\\ 
& = \exp(-\phi \mu_t(\Gamma(w_i)))\\
& \ge \exp(-\phi\lambda)
\end{align*}
where $\phi = \lambda/(1-\exp(-\lambda))$.
Hence we have
$$
\sum_{i=1}^m  \Prob[E_i]\Prob[F_i |E_i]\ge \exp(-\phi\lambda)\sum_{i=1}^m \Prob[E_i].
$$
But $\sum_{i=1}^m \Prob[E_i]$ is simply the probability that some vertex in $W$ beeps,
which is given by $1-\prod_{v \in W}(1-\mu_t(v))$.
Using Inequality~\eqref{exp} this value is at least $1-\exp(-\mu_t(W))$.  

Thus the probability that some vertex of $W$ is added to the independent set in round $t$ 
is at least
$$\exp(-\phi\lambda)\sum_{i=1}^m \Prob[E_i]\ge e^{-\phi\lambda}(1-e^{-\mu_t(W)}).$$
\hfill~
\qed
\end{proof}

\begin{proof}[of Theorem \ref{logrounds}]
Fix an arbitrary vertex $v$.
We shall show that, with failure probability $O(1/n^{k+1})$,
$v$ becomes inactive within $K_0 (k+1)\log n$ rounds, for a suitable choice of constant $K_0$.
Taking a union bound over all $n$ choices of $v$, it follows that with failure probability $O(1/n^k)$
every vertex becomes inactive and the algorithm terminates
within $K_0 (k+1)\log n$ rounds, which proves the theorem.

At each time step $t\ge1$, we partition the neighbourhood of $v$ into $\lambda$-light 
and $\lambda$-heavy vertices, for a suitable fixed choice of $\lambda$
\begin{align*}
L_t=L_t(v)=&\{x\in\Gamma(v) \mid x\ \text{is $\lambda$-light at step $t$}\} \\
H_t=H_t(v)=&\{x\in\Gamma(v) \mid x\ \text{is $\lambda$-heavy at step $t$}\}.
\end{align*}

We will follow the behaviour of $\mu_t(L_t)$ and $\mu_t(H_t)$ over time.

The idea of the argument is roughly as follows:  if $\mu_t(L_t)$
is large at many rounds, then by Lemma~\ref{light} it is very likely that some neighbour of $v$ will
be added to the independent set on one of these occasions, leading to $v$ becoming inactive.
If this does not happen, then  $\mu_t(L_t)$ must be small most of the time.
Now consider $H_t$.  Vertices that are $\lambda$-heavy at time $t$ are likely to hear beeps
and so drop in weight (as their signalling probability is reduced);
it will follow that with high probability $\mu_{t+1}(H_{t})$ is a constant factor smaller 
than $\mu_t(H_t)$ most of the time.
Now we look at the evolution of $\mu_t(\Gamma(v))$, the weight of the whole neighbourhood of $v$.
It may be large and increasing for some small fraction of the time,
but mostly it is either shrinking or else it is already small.  
It will follow that, for at least some fixed fraction of the time, $\mu_t(\Gamma(v))$ is small.
But this implies that, for at least some fixed fraction of the time, $v$ will not hear any beeps,
and hence $\mu_t(v)$ will be large for some fixed fraction of the time.
This implies that it is very likely that at some point in the sequence of rounds we are considering 
$v$ will beep and not hear any beeps, and so get added to the independent set.

To make this argument precise, we now define the following constants:
\begin{align*}
\label{constants}
r      & = 1 + (\log{f_2}/\log{f_1});\\
\lambda& = \log(32r(r+2)(f_2-f_1^{-1})/(f_2^{-1/r}-f_1^{-1}));\\
\phi   & = \lambda/(1-\exp(-\lambda));\\
\beta  & = 1/(4(r+2)\phi f_2);\\
\alpha & =(\beta/2)(f_2^{-1/r}-f_1^{-1})/(f_2-f_1^{-1});\\
K_0    & = (8r(r+2)) \max\{6,1/p_0,1/\log f_2,1/(e^{-\phi\lambda}(1-e^{-\alpha}))\};
\end{align*}
The values of these constants depend only on the fixed parameters $f_1$ and $f_2$ which bound the probability update factor $f$ used in the algorithm,
and on the initial minimum probability threshold $p_0$ (see Table~\ref{alg:MISlocal}). 
Note that $1 < f_1 \leq f_2$, so $r \geq 2$ and $\lambda > \log 256 > 5$.

To simplify the presentation, we also define
\begin{align*}
K      & = K_0 (k+1).
\end{align*}

At each round $t$, we consider the following four possible events:
\begin{enumerate}
\item[(E1)]  $\mu_t(L_t)\ge \alpha$

[`$\Gamma(v)$ has a significant weight of light neighbours']
\item[(E2)]  $\mu_t(L_t)< \alpha$ and $\mu_t(\Gamma(v))\le \beta$

[`$v$ is very light']
\item[(E3)]  $\mu_t(L_t)< \alpha$, $\mu_t(\Gamma(v))> \beta$
and $\mu_{t+1}(\Gamma(v))\le f_2^{-1/r}\mu_t(\Gamma(v))$

[`the neighbourhood of $v$ shrinks significantly in weight during round $t$']
\item[(E4)]  $\mu_t(L_t)< \alpha$, $\mu_t(\Gamma(v))> \beta$
and $\mu_{t+1}(\Gamma(v))>f_2^{-1/r}\mu_t(\Gamma(v))$

[`the neighbourhood of $v$ does not shrink significantly in weight during round $t$ (and may grow)']
\end{enumerate}
Exactly one of these events must occur in each round.  

We organize the rest of the proof as a series of claims.

\begin{subclaim}\label{claimE1limit}
With failure probability $O(1/n^{k+1})$, 
(E1) occurs at most $(K\log n)/(8r(r+2))$ times in the first $K\log n$ rounds.
\end{subclaim}

Each time that (E1) occurs, it follows from Lemma \ref{light} that 
with probability at least $ e^{-\phi\lambda}(1-e^{-\mu_t(L_t)}) \ge e^{-\phi\lambda}(1-e^{-\alpha})$
some vertex of $L_t$ is added to the independent set 
(and so $v$ becomes inactive and the process at $v$ terminates).
Let $\phi_1= e^{-\phi\lambda}(1-e^{-\alpha})$:
the probability that (E1) occurs
at least
$(K\log n)/(8r(r+2))$ times without $v$ becoming inactive is at most
$(1-\phi_1)^{(K\log n)/(8r(r+2))}$ which is
at most $$\exp(-(\phi_1K_0/(8r(r+2)))(k+1)\log n).$$
By our choice of $K_0$, we have $K_0 \geq (8r(r+2))/\phi_1$,
so this probability is at most $\exp(-(k+1)\log n)=n^{-(k+1)}$.  
This proves Claim~\ref{claimE1limit}.

The bad event for us will be (E4), so let us bound the probability that (E4) occurs.
\begin{subclaim}\label{claimE4prob}
At each round $t$, the probability that (E4) occurs is at most $1/(16r(r+2))$.
\end{subclaim}
If (E4) can occur, then we must have that $\mu_t(L_t)< \alpha$ and $\mu_t(\Gamma(v))> \beta$. 
For any $x\in H_t$, there are two cases to consider - the first is that the total weight of all its
neighbouring vertices is greater than $\lambda$; the second is that at least one of its neighbouring
vertices individually has weight more than $1-\exp(-\lambda)$.

In the first case, using Inequality~\eqref{exp}, the probability that no neighbour of $x$ beeps 
in round $t$ is at most $\exp(-\mu_t(\Gamma(x)))$ which is bounded by $\exp(-\lambda)$. 
In the second case, the probability that no neighbour of $x$ beeps in round $t$ 
is still at most $\exp(-\lambda)$. 
Thus, for any $x \in H_t$ we have shown that the probability that no neighbour of $x$ beeps 
in round $t$ is at most $\exp(-\lambda)$. 

Let $H_t^0$ be the set of vertices in $H_t$ that do not hear a beep in round $t$,
and let $H_t^1=H_t\setminus H_t^0$ be the remaining vertices in $H_t$ that do hear a beep. Then
$
\E[\mu_t(H_t^0)]\le \exp(-\lambda)\mu_t(H_t),
$
and so by Markov's inequality
\begin{equation}\label{mi}
\Prob\left[\mu_t(H_t^0) \ge 16r(r+2)\exp(-\lambda)\mu_t(H_t)\right]\le 1/(16r(r+2)).
\end{equation}
Now all vertices in $H_t^1$ decrease their weight by a factor of at least $f_1$, 
while vertices in $L_t$ and $H_t^0$ may either decrease or increase their weight 
(additionally, some weights may get set to 0 if vertices become inactive).  So
\begin{align*}
\mu_{t+1}(\Gamma(v))
& \le \frac{1}{f_1}\mu_t(H_t^1) + f_2\mu_t(H_t^0) +f_2\mu_t(L_t)\\
& =\frac1{f_1}\mu_t(\Gamma(v))\\
& ~~~~+ (f_2-\frac1{f_1})\mu_t(H_t^0) +(f_2-\frac1{f_1})\mu_t(L_t)
\end{align*}
It follows from Inequality~\eqref{mi} that, with probability at least $1-1/(16r(r+2))$,
\begin{align*}
\mu_{t+1}(\Gamma(v))
& \le \frac1{f_1}\mu_t(\Gamma(v)) + (f_2-\frac1{f_1}) 16r(r+2)e^{-\lambda}\mu_t(H_t)\\
& \hspace{1cm} +(f_2-\frac1{f_1})\mu_t(L_t)\\
& \le f_2^{-1/r}\mu_t(\Gamma(v)),
\end{align*}
where the final inequality follows from our
choice of $\lambda$, which gives
\begin{align*}
(f_2-f_1^{-1}) & 16r(r+2)e^{-\lambda}\mu_t(H_t)\\
               & \leq (f_2-f_1^{-1}) 16r(r+2)e^{-\lambda}\mu_t(\Gamma(v))\\
               & \leq 1/2(f_2^{-1/r}-f_1^{-1})\mu_t(\Gamma(v)),
\end{align*}
and our choice of $\alpha$ and $\beta$, 
because we are assuming that $\mu_t(L_t)< \alpha$ and $\mu_t(\Gamma(v))> \beta$, so we have 
\begin{align*}
(f_2-f_1^{-1})\mu_t(L_t)  
& < (f_2-f_1^{-1})\alpha \\
& = 1/2(f_2^{-1/r}-f_1^{-1})\beta\\ 
& < 1/2(f_2^{-1/r}-f_1^{-1})\mu_t(\Gamma(v)).
\end{align*}
Thus the probability that (E4) does occur is bounded above by $1/(16r(r+2))$.  
This proves Claim~\ref{claimE4prob}.

\begin{subclaim}\label{claimE4limit}
With failure probability $O(1/n^{k+1})$, (E4) occurs at most $(K\log n)/(8r(r+2))$ times
in the first $K\log n$ rounds.
\end{subclaim}
At each round, the probability of (E4) depends on the past history of the process.
However, by Claim~\ref{claimE4prob}, it is always at most $1/(16r(r+2))$, and so
we can couple occurrences of (E4) with a sequence of independent events each occurring with probability
$1/(16r(r+2))$.  It follows that the number of occurrences of (E4)
in the first $K\log n$ rounds is
stochastically dominated by a binomial random variable $X$ with parameters $K\log n$ and $1/(16r(r+2))$.
The probability that (E4) occurs more than $(K\log n)/(8r(r+2))$ times is therefore, 
by \eqref{bigup}, at most $$\Prob[X>2\E X] \le \exp(-\E X/3) \le \exp(-(K\log n)/(48r(r+2))).$$
By our choice of $K_0$, we have $K_0 \ge 48r(r+2)$, so this probability is $O(n^{-(k+1)})$, 
which proves Claim~\ref{claimE4limit}.

From Claim~\ref{claimE1limit} and Claim~\ref{claimE4limit},
we conclude that with failure probability $O(n^{-(k+1)})$, (E1) and (E4) altogether occur at most
$K(\log n)/(4r(r+2))$ times in the first $K\log n$ rounds. 
We next show that, with small failure probability, $\mu_t(\Gamma(v))$ is small most of the time.

\begin{subclaim}\label{runclaim}
With failure probability $O(1/n^{k+1})$,
$\mu_t(\Gamma(v))>f_2\beta$ at most $(K\log n)/(2(r+2))$ times 
in the first $K\log n$ rounds.
\end{subclaim}

Let $T$ be the set of rounds $t \geq 1$ at which $\mu_t(\Gamma(v)) > f_2\beta$.
We decompose $T$ into (maximal) intervals of integers, say as $T_1\cup \dots \cup T_m$.
Let $T_i=[s_i,t_i]$ be one of these intervals.
We colour each integer $t\in T_i$ {\em red} if (E1) or (E4) occurred at the previous round, 
and {\em blue} if (E3) occurred 
(note that  (E2) cannot occur, as $\mu_{t-1}(\Gamma(v))\ge\mu_t(\Gamma(v))/f_2>\beta$).
By the definition of (E3), we have $\mu_t(\Gamma(v))\le f_2^{-1/r}\mu_{t-1}(\Gamma(v))$ at blue rounds, 
and we have $\mu_t(\Gamma(v))\le f_2\mu_{t-1}(\Gamma(v))$ otherwise.
Let $r_i$ be the number of red elements in $T_i$ and $b_i$ the number of blue elements.
It follows that 
$$\mu_{t_i}(\Gamma(v))\le \mu_{s_i-1}(\Gamma(v)) \cdot {f_2}^{r_i-b_i/r}.$$
Since $\mu_{t_i}(\Gamma(v))>f_2\beta$ it follows that
\begin{equation*}
r_i>\frac{1}{r}b_i + \log_{f_2} f_2\beta - \log_{f_2} (\mu_{s_i-1}(\Gamma(v))).
\end{equation*}
However, $\mu_{s_i-1}(\Gamma(v)) \leq f_2 \beta$ in all cases where $s_i > 1$,
and $\mu_0(\Gamma(v)) < n$. 
Summing over $i$, we see that 
$$
\sum_{i=1}^{m} r_i > \frac{1}{r}\sum_{i=1}^{m} b_i + \log_{f_2} f_2\beta - \log_{f_2} n
$$
But red rounds correspond to events (E1) and (E4), and we have already shown
in Claims~\ref{claimE1limit} and~\ref{claimE4limit}
 that these occur at most $(K\log n)/(4r(r+2))$ times altogether in the first $K\log n$ rounds. 
Hence the total number of rounds in $T$, both red and blue, is less than  
$(K\log n)/(4r(r+2)) + (K\log n)/(4(r+2)) + r \log n/\log f_2$.
By our choice of $K_0$, this is less than $(K\log n)/(2(r+2))$.
This proves Claim~\ref{runclaim}.

\begin{subclaim}\label{nohearclaim}
With failure probability $O(1/n^{k+1})$, $v$ hears a beep at most $(K\log n)/(r+2)$ times 
in the first $K\log n$ rounds.
\end{subclaim}

By our choice of $\beta$, we have that 
$f_2\beta = (1 - e^{-\lambda})/(4 (r+2) \lambda) < (1 - e^{-\lambda})$.
Hence we may apply Inequality~\eqref{exp2}, to show that when $\mu_t(\Gamma(v)) \le f_2\beta$ 
the probability that $v$ hears no beep is 
\begin{align*}
\prod_{x \in \Gamma(v)}(1-\mu_t(x)) 
& \ge \prod_{x \in \Gamma(v)}\exp(-\phi\mu_t(x))\\
& \ge \exp(-\phi\mu_t(\Gamma(v)))\\
& \ge \exp(-\phi f_2\beta) \ge 1-\phi f_2\beta,
\end{align*}
and so $v$ hears a beep with probability at most $\phi f_2\beta$ which equals $1/(4(r+2))$.
Using 
\eqref{bigup}, this implies that with failure probability $O(n^{-(k+1)})$
there are at most $(K\log n)/(2(r+2))$ rounds among the first $K\log n$ 
at which $\mu_t(\Gamma(v)) \le f_2\beta$ and $v$ hears a beep. 
By Claim~\ref{runclaim}, with the same failure probability,
there are also at most $(K\log n)/(2(r+2))$ rounds at which $\mu_t(\Gamma(v)) > f_2\beta$
(and $v$ might hear a beep at any of these steps).
It follows that, with failure probability $O(n^{-(k+1)})$,
$v$ hears a beep at most $(K\log n)/(r+2)$ times
in the first $K\log n$ rounds.
This proves Claim~\ref{nohearclaim}.

\begin{subclaim}\label{claimterminate}
With failure probability $O(1/n^{k+1})$,
$v$ becomes inactive during the first $K\log n$ rounds.
\end{subclaim}

From the previous claim, we may assume that $v$ hears a beep on at most $K\log n/(r+2)$
occasions during the first $K\log n$ rounds.
On these occasions it decreases its local probability value $p$ by a factor of at most $f_2$.
We shall refer to these as red steps.

Hence there are at least $\frac{(r+1)}{(r+2)} K \log n$ rounds during the first $K\log n$ 
rounds where $v$ does not hear a beep, so it either terminates, 
or increases its local probability value $p$ by a factor of at least $f_1$, 
or else increases $p$ to 1.
We shall refer to these as blue rounds.
Note that if $v$ beeps in a blue round then it will terminate in that round. 
Hence a blue round where the value of $p$ increases to 1 must be immediately followed by a red 
round, or a blue round where $v$ terminates.

Now, by our choice of $r$, $f_1^r > f_2$. 
This means that there must be 
at least $\frac{1}{(r+2)} K \log n$ blue rounds during the first $K\log n$ rounds
where either $v$ has terminated, or else the local probability value $p$ at $v$ is at least 
as high as the initial value, $p_0$.

The probability that $v$ will terminate at each of these blue rounds is at least $p_0$,
so the probability that $v$ remains active throughout all these blue rounds is at most
$(1-p_0)^{K \log n/(r+2)}$. 
Using Inequality~\eqref{exp}, this means that the probability that $v$ remains active throughout these
rounds is at most 
$$\exp(-p_0 K \log n/(r+2)).$$
By our choice of $K_0$, we have $K_0 > (r+2)/p_0$, so this is $O(n^{-(k+1)})$.
Hence $v$ terminates with a failure probability that is $O(n^{-(k+1)})$.

This proves Claim~\ref{claimterminate}, and completes the proof of Theorem~\ref{logrounds}.
\hfill~
\qed
\end{proof}

\begin{corollary}\label{logexpected}
The expected number of rounds taken by the algorithm in Table~\ref{alg:MISlocal} 
on any graph with $n$ nodes is $O(\log n)$.
\end{corollary}
\begin{proof}
Let $T$ be the total number of rounds taken by the algorithm and let
$T' = \lceil T/(K_0 \log n) \rceil$, where $K_0$ is the constant identified
in Theorem~\ref{logrounds}.

By Theorem~\ref{logrounds}, we have that, for any $k \geq 1$,
$$\Prob[T' > k+1] \leq c'/n^{k}$$
for some constant $c'$.
Hence $\E[T'] = \sum_{k\ge 1}\Prob[T' \ge k] = O(1)$.
\hfill~
\qed
\end{proof}


\subsection{Expected Number of Signals}

In this section we will show that the expected number of
times that each node signals is bounded by a constant. 
Hence the expected bit complexity per node for this algorithm
{\em does not increase at all} with the number of nodes.

\begin{theorem}
\label{thm:expectedbeeps}
The expected total number of signals broadcast by any node executing the algorithm in Table~\ref{alg:MISlocal} is $O(1)$.
\end{theorem}
\begin{proof}
Let $v$ be a node executing the algorithm in Table~\ref{alg:MISlocal},
and consider the whole sequence of rounds until $v$ becomes inactive.

Once again we will refer to sending a signal in the first exchange of the algorithm in Table~\ref{alg:MISlocal}
as ``beeping", and receiving such a signal from a neighbour as ``hearing a beep". 

During each round, one of the following 3 things happens in the first exchange:
\begin{description}
\item[Case 1] $v$ hears a beep from its neighbours,
and so decreases its probability of beeping by a factor of $f$ (from $p_t$ to $\frac{1}{f}p_t$).
We will call these {\em red} rounds.
\item[Case 2] $v$ hears no beep from its neighbours,
and so increases its probability of beeping by a factor of $f$ (from $p_t$ to $fp_t$).
We will call these {\em blue} rounds.
\item[Case 3] $v$ hears no beep from its neighbours,
and so increases its probability of beeping to $1$.
We will call these {\em dark blue} rounds.
\end{description}

If $v$ beeps in any blue or dark blue round then it joins the MIS and becomes inactive, 
so the total number of beeps at such rounds is at most one, at the final round in the sequence.
Hence we only need to consider the expected number of beeps at red rounds.

Consider first those red rounds (if any) where the value of $p_t$ is 
at its lowest point in the sequence so far.
The expected number of times that $v$ beeps during this subsequence of rounds 
is bounded by  
$p_1+\frac{p_1}{f_1}+\frac{p_1}{f_1^2}+\frac{p_1}{f_1^3}\cdots 
\leq 
\frac{f_1}{f_1-1}p_1
\leq \frac{f_1}{f_1-1}$.

At all of the remaining red rounds the value of $p_t$ is not at its lowest point so far,
so it was lower at some previous blue round. 
Hence each of these red rounds can be associated with a corresponding earlier blue round: 
the most recent blue round where the value of $p_t$ was lower.
We now define the constant $r = \lceil \log f_2/\log f_1 \rceil$,
where $f_1$ and $f_2$ denote respectively the lower bound and upper bound of $f$ as given in Table~\ref{alg:MISlocal}. 
Note that $f_1^r \geq f_2$. 
Since blue rounds increase the value of $p_t$ by at most a factor of $f_2$,
and red rounds decrease the value of $p_t$ by at least a factor of $f_1$,
it follows that each blue round will be associated with at most $r$ red rounds.

Hence we have partitioned the remaining red rounds into groups of at most $r$ red rounds, 
each associated with a single (earlier) blue round.
We now consider these groups of at most $r+1$ rounds, ordered by the position of the initial blue round.
For any such group, if the probability of beeping at the blue round is $p$, 
then the probability of beeping at any of the associated red rounds is at most $f_2 p$.
Hence the conditional probability of beeping at the initial blue round, given that a beep occurs
somewhere in this group of rounds, is at least $1/(1+rf_2)$.
Hence if we consider the subsequence of groups where at least one beep occurs,
the expected number of such groups before a beep occurs at a blue round is at most $rf_2$
(expected number of failures before the first success in a geometric distribution).
Since each group can contribute at most $r$ beeps,
the expected number of beeps added in these groups
before terminating is at most $r (rf_2)$.

We have shown that the expected number of times that $v$ beeps is at most
$1 + \frac{f_1}{f_1-1} +(\lceil \log {f_2}/\log f_1 \rceil)^2 f_2$, which proves the result.
\qed
\end{proof}

\section{Algorithm for Distributed Greedy Colouring}
\label{secGC}

Our new algorithm for distributed greedy colouring (see Table~\ref{alg:GClocal}) 
is similar to our new distributed MIS selection algorithm. 
At each round, each node may choose, with a certain probability $p$, 
to broadcast its first available colour to all its neighbours, 
indicating that it wishes to use that colour. 
If two neighbouring nodes broadcast the same colour in the same round, 
then they will both abandon choosing that colour in that round. 
On the other hand, if a node broadcasts a colour and none of its neighbouring nodes 
broadcast the same colour in that round, then it is successfully coloured, 
and will notify all its neighbouring nodes that they are forbidden to use that colour. 

As in our MIS selection algorithm, the way that the probability values $p$
are chosen is inspired by the positive feedback mechanisms that control cellular
processes.
The value of $p$ is initialised to some arbitrary value at each node
(above some strictly positive fixed threshold value, $p_0$). 
These values are then independently updated at each node in each round 
using feedback from neighbouring nodes. 
The value of $p$ is decreased at a node 
whenever one or more neighbouring nodes broadcast the same colour, 
and is increased whenever no neighbouring node broadcasts the same colour. 
As in our MIS selection algorithm, we allow each increase or decrease to be by some arbitrary factor $f$,
which may vary at each round, but is always bounded by the global parameters $f_1$ and $f_2$ 
(with $1<f_1\le f_2$).

\begin{table*}[t]
\caption{The algorithm for distributed greedy colouring at each node
\label{alg:GClocal}}{
\centering
\framebox{
\begin{minipage}[c]{0.85\linewidth}
\begin{tabbing}
{\bf Global constants:} \= $p_0$ : lower bound on initial probability value; \\
         \> $f_1,f_2$ : lower and upper bounds on change factor for probability value.\\[0.1cm]
{\bf Local variables:} \> $p$ : local probability value, initialised to some value in $[p_0,1]$;\\
         \> $f$ : change factor for probability value, chosen arbitrarily in $[f_1,f_2]$;\\
         \> {\sc Trying} : Boolean flag, initialised to {\sc False};\\
         \> $S$ : Set of forbidden colours, i.e., those taken by neighbours, initialised to $\emptyset$.
\end{tabbing}
\algsetup{
linenosize=\tiny,
linenodelimiter=.,
indent = 3em
}
\renewcommand{\algorithmiccomment}[1]{\hfill {\sc \small #1}}
\begin{algorithmic}[1]
\WHILE{active, at each time step}
\vspace*{0.2cm}
\STATE *FIRST EXCHANGE*
\STATE Choose the smallest available colour $c$ that is not in $S$;
\STATE With probability $p$, 
set {\sc Trying} $\leftarrow$ {\sc True} and {\bf send} $c$ to all neighbours;
\STATE {\bf Receive} any colour signals sent by neighbours;
\STATE Set $f$ to some value in the interval $[f_1,f_2]$;
\IF{any neighbour sent colour $c$}
   \STATE {\sc Trying} $\leftarrow$ {\sc False} and $p\leftarrow p/f$ (decrease $p$)  
\ELSE
   \STATE $p\leftarrow \min\{fp,1\}$ (increase $p$)
\ENDIF

\vspace*{0.2cm}
\STATE *SECOND EXCHANGE*
\IF{{\sc Trying}}
    \STATE {\bf Send} $c$ to all neighbours;
    \STATE Assign colour $c$ to this node and terminate (become inactive).
\ENDIF
\STATE {\bf Receive} any colour signals sent by neighbours and add all distinct colours received to $S$.
\ENDWHILE
\vspace*{0.2cm}
\end{algorithmic}
\end{minipage}
}
}
\end{table*}

The correctness of the algorithm in Table~\ref{alg:GClocal} follows easily from the two facts below:
\begin{description}
\item[Fact 1] No two nodes that are assigned the same colour in the same round are adjacent. 
\item[Fact 2] The colour assigned to any node is the smallest colour that is different from 
all colours previously assigned to neighbouring nodes. 
\end{description}
Thus, if the algorithm in Table~\ref{alg:GClocal} is run on the nodes of any graph $G$, and all nodes
become inactive, then the colour assigned to each of the nodes defines a greedy colouring of $G$. 

Our analysis of the distributed greedy colouring algorithm shown in Table~\ref{alg:GClocal} 
is very similar to the analysis for the MIS selection algorithm given in Section~\ref{sec:localtime}.

\begin{theorem}\label{GCrounds}
For any fixed values of $p_0 > 0$, and $1 < f_1 \leq f_2$, 
there is a constant $K_0$ and a constant $r$ such that the following holds:
For any graph $G$ with $n$ vertices and maximum degree $\Delta$, and any $k\ge1$,
the algorithm defined in Table~\ref{alg:GClocal} terminates in 
at most $8r(r+2)\Delta + K_0(k+1)\log n$ rounds,
with probability at least $1-O(1/n^k)$.
\end{theorem}

As in Section~\ref{sec:localtime}, we refer to broadcasting any colour $c$ in the first exchange as
``beeping", and receiving {\em the same} colour $c$ from a neighbour in that exchange as
``hearing a beep". 
As before, for any vertex $v$, at any time step $t$, we define the measure $\mu_t(v)$, 
called the ``weight'' of $v$, to be the probability that $v$ beeps in round $t$. 

The set of neighbours of a vertex $v$ which are competing for the same colour as $v$
in any round will be called the {\em homogeneous} neighbours of $v$, 
and will be denoted by $\Gamma^{(h)}(v)$.
We adapt Definition~\ref{def:lambdalightMIS} to refer to homogeneous neighbours only, as follows.
\begin{definition}\label{def:lighth}
For any $\lambda > 0$, 
a vertex $v$ will be called {\em $\lambda$-light$^{(h)}$} at round $t$  
if $\mu_t(\Gamma^{(h)}(v))\le\lambda$ 
and every homogeneous neighbour of $v$ has weight at most $1-\exp(-\lambda)$;
otherwise, vertex $v$ is called {\em $\lambda$-heavy$^{(h)}$}.
\end{definition}
As in Section~\ref{sec:localtime}, we first establish a lower bound on the probability 
that at least one vertex in a set of $\lambda$-light$^{(h)}$ vertices will be coloured at each round.
\begin{lemma}\label{GClight}
Let $W$ be a set of vertices that are $\lambda$-light$^{(h)}$ at round $t$. 
The probability that at least one vertex in $W$ gets coloured in round $t$ is at least
$e^{-\phi\lambda}(1-e^{-\mu_t(W)})$ where $\phi = \lambda/(1-\exp(-\lambda))$.
\end{lemma}
\begin{proof}
Identical to the proof of Lemma~\ref{light}.
\hfill~
\qed
\end{proof}
\begin{proof}[of Theorem \ref{GCrounds}]
Fix an arbitrary vertex $v$.
We use essentially the same argument as in the proof of Theorem~\ref{logrounds},
partitioning the neighbourhood of $v$ into $\lambda$-light$^{(h)}$ and $\lambda$-heavy$^{(h)}$
vertices, and following the progress of these sets over time.
Note that we consider the entire neighbourhood of $v$, not just the homogeneous neighbours.
We show that the weight of this entire neighbourhood is small for at least a fixed fraction of the 
time, and hence $v$ fails to receive any colour signals for at least a fixed fraction of the time.

We define the same constants, and the same events $(E1)$ to $(E4)$ 
(see page 14). 
However, in this proof we will need to allow for the possibility that one or more neighbours
of $v$ are successfully coloured at any time step, which can happen up to $\Delta$ times,
and does not immediately force $v$ to become inactive.


Each time that (E1) occurs, it follows from Lemma \ref{GClight} that 
with probability at least $e^{-\phi\lambda}(1-e^{-\alpha})$
some $\lambda$-light$^{(h)}$ neighbour of $v$ will be coloured.
Let $\phi_1= e^{-\phi\lambda}(1-e^{-\alpha})$.
As in the proof of Claim~\ref{claimE1limit},
the probability that there are $(K\log n)/(8r(r+2))$ occurrences of $(E1)$
where no neighbour of $v$ is coloured is at most
$(1-\phi_1)^{(K\log n)/(8r(r+2))}\le \exp(-(\phi_1K_0/(8r(r+2)))(k+1)\log n)$.
By our choice of $K_0$, we have $K_0 \geq (8r(r+2))/\phi_1$,
so this probability is at most $\exp(-(k+1)\log n)=n^{-(k+1)}$.  

Hence with failure probability $O(1/n^{k+1})$, the {\em total} number of 
occurrences of $(E1)$ 
is at most $\Delta + (K\log n)/(8r(r+2))$.

Since we have a weaker upper bound on the number of occurrences of $(E1)$,
we will need to consider a longer sequence of rounds overall.
In fact, we will consider the first $8r(r+2)\Delta + K\log n$ rounds.

Following exactly the same arguments as in the proof of 
Claims~\ref{claimE4prob} and \ref{claimE4limit},
we obtain that with failure probability $O(1/n^{k+1})$, 
(E4) occurs at most $\Delta + (K\log n)/(8r(r+2))$ times 
in the first $8r(r+2)\Delta + K\log n$ rounds.

Next, following the argument used to prove Claim~\ref{runclaim},
but using the weaker bound of $2\Delta + (K\log n)/(4r(r+2))$ for the total number of red rounds,
we obtain that with failure probability $O(1/n^{k+1})$,
$\mu_t(\Gamma(v))>f_2\beta$ at most $4r\Delta + (K\log n)/(2(r+2))$ times 
in the first $8r(r+2)\Delta + K\log n$ rounds.

Now, using the same argument as in Claim~\ref{nohearclaim} we can show that
with failure probability $O(1/n^{k+1})$, the fraction of rounds where 
$v$ hears a beep (or in fact receives any colour signal in the first exchange) 
during the first $8r(r+2)\Delta + K\log n$ rounds is at most $1/(r+2)$.

Finally, using the same argument as in Claim~\ref{claimterminate}, this implies that
with failure probability $O(1/n^{k+1})$,
$v$ becomes inactive during the first $8r(r+2)\Delta + K\log n$ rounds.

Taking a union bound over all possible choices of $v$, as before, gives the result.
\hfill~
\qed
\end{proof}
\begin{corollary}\label{GClogexpected}
The expected number of rounds taken by the algorithm in Table~\ref{alg:GClocal} 
on any graph with $n$ nodes is $O(\Delta+\log n)$.
\end{corollary}
Finally, the proof of Theorem~\ref{thm:expectedbeeps} considers only an individual node
and does not take into account whether 
the neighbours of a node become inactive or not, so this proof 
applies equally well to our distributed colouring algorithm, giving the following result.
\begin{theorem}
\label{thm:GCexpectedbeeps}
The expected total number of beeps broadcast by any node executing the algorithm in Table~\ref{alg:GClocal} is $O(1)$.
\end{theorem}


\begin{thebibliography}{10}
\providecommand{\url}[1]{{#1}}
\providecommand{\urlprefix}{URL }
\expandafter\ifx\csname urlstyle\endcsname\relax
  \providecommand{\doi}[1]{DOI~\discretionary{}{}{}#1}\else
  \providecommand{\doi}{DOI~\discretionary{}{}{}\begingroup
  \urlstyle{rm}\Url}\fi

\bibitem{Afek2011b}
Afek, Y., Alon, N., Bar-Joseph, Z., Cornejo, A., Haeupler, B., Kuhn, F.:
  Beeping a maximal independent set.
\newblock In: Proceedings of the 25th International Conference on Distributed
  Computing, DISC'11, pp. 32--50. Springer-Verlag (2011)

\bibitem{Afek2011a}
Afek, Y., Alon, N., Barad, O., Hornstein, E., Barkai, N., Bar-Joseph, Z.: A
  biological solution to a fundamental distributed computing problem.
\newblock Science \textbf{331}(6014), 183--185 (2011)

\bibitem{Alon1986a}
Alon, N., Babai, L., Itai, A.: A fast and simple randomized parallel algorithm
  for the maximal independent set problem.
\newblock Journal of Algorithms \textbf{7}, 567--583 (1986)

\bibitem{Barenboim2009d}
Barenboim, L., Elkin, M.: Distributed ($\delta+1$)-coloring in linear (in
  $\delta$) time.
\newblock In: Proceedings of the 41st Annual ACM Symposium on Theory of
  Computing, STOC '09, pp. 111--120. ACM, New York, NY, USA (2009)

\bibitem{Barenboim2010s}
Barenboim, L., Elkin, M.: {Sublogarithmic distributed MIS algorithm for sparse
  graphs using Nash-Williams decomposition}.
\newblock Distributed Computing \textbf{22}(5-6), 363--379 (2010)

\bibitem{Barenboim2011d}
Barenboim, L., Elkin, M.: Deterministic distributed vertex coloring in
  polylogarithmic time.
\newblock Journal of the ACM \textbf{58}(5), 23:1--23:25 (2011)

\bibitem{Barenboim2012t}
Barenboim, L., Elkin, M., Pettie, S., Schneider, J.: The locality of
  distributed symmetry breaking.
\newblock In: Proceedings of the 2012 IEEE 53rd Annual Symposium on Foundations
  of Computer Science, FOCS '12, pp. 321--330. IEEE Computer Society,
  Washington, DC, USA (2012)

\bibitem{Bray2006n}
Bray, S.J.: Notch signalling: a simple pathway becomes complex.
\newblock Nature reviews Molecular cell biology \textbf{7}(9), 678--689 (2006)

\bibitem{Chaudhuri1987a}
Chaudhuri, P.: Algorithms for some graph problems on a distributed
  computational model.
\newblock Information Sciences \textbf{43}(3), 205--228 (1987)

\bibitem{Collier1996p}
Collier, J.R., Monk, N.A., Maini, P.K., Lewis, J.H.: Pattern formation by
  lateral inhibition with feedback: a mathematical model of delta-notch
  intercellular signalling.
\newblock Journal of Theoretical Biology \textbf{183}(4), 429--446 (1996)

\bibitem{Cornejo2010d}
Cornejo, A., Kuhn, F.: Deploying wireless networks with beeps.
\newblock In: Proceedings of the 24th International Conference on Distributed
  Computing, DISC'10, pp. 148--162. Springer-Verlag, Berlin, Heidelberg (2010)

\bibitem{Emek2013s}
Emek, Y., Wattenhofer, R.: Stone age distributed computing.
\newblock In: Proceedings of the 2013 ACM Symposium on Principles of
  Distributed Computing, PODC '13, pp. 137--146. ACM, New York, NY, USA (2013)

\bibitem{Gavoille2009o}
Gavoille, C., Klasing, R., Kosowski, A., Kuszner, {\L}., Navarra, A.: On the
  complexity of distributed graph coloring with local minimality constraints.
\newblock Networks \textbf{54}(1), 12--19 (2009)

\bibitem{Gebremedhin2003g}
Gebremedhin, A.H., Lassous, I.G., Gustedt, J., Telle, J.A.: Graph coloring on
  coarse grained multicomputers.
\newblock Discrete Applied Mathematics \textbf{131}(1), 179--198 (2003)

\bibitem{Goldberg1988p}
Goldberg, A.V., Plotkin, S.A., Shannon, G.E.: Parallel symmetry-breaking in
  sparse graphs.
\newblock SIAM Journal on Discrete Mathematics \textbf{1}(4), 434--446 (1988)

\bibitem{Greenlaw1995l}
Greenlaw, R., Hoover, H.J., Ruzzo, W.L. (eds.): Limits to Parallel Computation:
  {P}-completeness Theory.
\newblock Oxford University Press, Inc., New York, NY, USA (1995)

\bibitem{Grundy1939m}
Grundy, P.: Mathematics and games.
\newblock Eureka \textbf{2}, 6--8 (1939)

\bibitem{Halpern1990k}
Halpern, J.Y., Moses, Y.: Knowledge and common knowledge in a distributed
  environment.
\newblock Journal of the ACM \textbf{37}(3), 549--587 (1990)

\bibitem{Hansen2004d}
Hansen, J., Kubale, M., Kuszner, {\L}., Nadolski, A.: Distributed largest-first
  algorithm for graph coloring.
\newblock In: M.~Danelutto, M.~Vanneschi, D.~Laforenza (eds.) Euro-Par 2004
  Parallel Processing, \emph{Lecture Notes in Computer Science}, vol. 3149, pp.
  804--811. Springer Berlin Heidelberg (2004)

\bibitem{Hedetniemi2003l}
Hedetniemi, S.T., Jacobs, D.P., Srimani, P.K.: Linear time self-stabilizing
  colorings.
\newblock Information Processing Letters \textbf{87}(5), 251--255 (2003)

\bibitem{Itai1990s}
Itai, A., Rodeh, M.: Symmetry breaking in distributed networks.
\newblock Information and Computation \textbf{88}(1), 60--87 (1990)

\bibitem{Johansson1999s}
Johansson, O.: Simple distributed {$\Delta+1$}-coloring of graphs.
\newblock Information Processing Letters \textbf{70}(5), 229 -- 232 (1999)

\bibitem{Karp1972r}
Karp, R.M.: Reducibility among combinatorial problems.
\newblock In: R.E. Miller, J.W. Thatcher (eds.) Complexity of Computer
  Computations, The IBM Research Symposia Series, pp. 85--103. Plenum Press,
  New York (1972)

\bibitem{Karp1985a}
Karp, R.M., Wigderson, A.: A fast parallel algorithm for the maximal
  independent set problem.
\newblock Journal of the ACM \textbf{32}(4), 762--773 (1985)

\bibitem{Kroeker2011b}
Kroeker, K.L.: Biology-inspired networking.
\newblock Communications of the ACM \textbf{54}, 11--13 (2011)

\bibitem{Kuhn2009w}
Kuhn, F.: Weak graph colorings: distributed algorithms and applications.
\newblock In: Proceedings of the 21st Annual Symposium on Parallelism in
  Algorithms and Architectures, SPAA '09, pp. 138--144. ACM, New York, NY, USA
  (2009)

\bibitem{Kuhn2005f}
{Kuhn}, F., {Moscibroda}, T., {Nieberg}, T., {Wattenhofer}, R.: Fast
  deterministic distributed maximal independent set computation on
  growth-bounded graphs.
\newblock In: P.~{Fraigniaud} (ed.) Distributed Computing: 19th International
  Conference, DISC 2005, \emph{Lecture Notes in Computer Science}, vol. 3724,
  pp. 273--283. Springer (2005)

\bibitem{Kuhn2004w}
Kuhn, F., Moscibroda, T., Wattenhofer, R.: What cannot be computed locally!
\newblock In: Proceedings of the 23rd Annual ACM Symposium on Principles of
  Distributed Computing, PODC '04, pp. 300--309. ACM, New York, NY, USA (2004)

\bibitem{Kuhn2006t}
Kuhn, F., Moscibroda, T., Wattenhofer, R.: The price of being near-sighted.
\newblock In: Proceedings of the Seventeenth Annual ACM-SIAM Symposium on
  Discrete Algorithm, SODA '06, pp. 980--989. New York, NY, USA (2006)

\bibitem{Kuhn2010}
Kuhn, F., Moscibroda, T., Wattenhofer, R.: Local computation: Lower and upper
  bounds.
\newblock CoRR \textbf{abs/1011.5470} (2010).
\newblock \urlprefix\url{http://arxiv.org/abs/1011.5470}

\bibitem{Kuhn2006o}
Kuhn, F., Wattenhofer, R.: On the complexity of distributed graph coloring.
\newblock In: Proceedings of the 25th Annual {ACM} Symposium on Principles of
  Distributed Computing, PODC '06, pp. 7--15. New York, NY, USA (2006)

\bibitem{Lenzen2012d}
Lenzen, C., Wattenhofer, R.: Distributed algorithms for sensor networks.
\newblock Philosophical Transactions of the Royal Society A: Mathematical,
  Physical and Engineering Sciences \textbf{370}(1958), 11--26 (2012)

\bibitem{Linial1986l}
Linial, N.: Legal coloring of graphs.
\newblock Combinatorica \textbf{6}(1), 49--54 (1986)

\bibitem{Linial1987d}
Linial, N.: Distributive graph algorithms -- global solutions from local data.
\newblock In: Proceedings of the 28th Annual Symposium on Foundations of
  Computer Science, SFCS '87, pp. 331--335. IEEE Computer Society, Washington,
  DC, USA (1987)

\bibitem{Linial1992l}
Linial, N.: Locality in distributed graph algorithms.
\newblock SIAM Journal on Computing \textbf{21}(1), 193--201 (1992)

\bibitem{Luby1986a}
Luby, M.: A simple parallel algorithm for the maximal independent set problem.
\newblock SIAM Journal on Computing \textbf{15}(4), 1036--1053 (1986)

\bibitem{Lynch1996d}
Lynch, N.A.: Distributed Algorithms.
\newblock Morgan Kaufmann Publishers Inc., San Francisco, CA, USA (1996)

\bibitem{Maan2012a}
Maan, V., Purohit, G.N.: A distributed approach for frequency allocation using
  graph coloring in mobile networks.
\newblock International Journal of Computer Applications \textbf{58}(6), 9--13
  (2012).
\newblock Published by Foundation of Computer Science, New York, USA

\bibitem{Metivier2010a}
M{\'e}tivier, Y., Robson, J.M., Saheb-Djahromi, N., Zemmari, A.: About
  randomised distributed graph colouring and graph partition algorithms.
\newblock Information and Computation \textbf{208}(11), 1296--1304 (2010)

\bibitem{Metivier2011a}
M{\'e}tivier, Y., Robson, J.M., Saheb-Djahromi, N., Zemmari, A.: An optimal bit
  complexity randomized distributed {MIS} algorithm.
\newblock Distributed Computing \textbf{23}(5-6), 331--340 (2011)

\bibitem{Moscibroda2005m}
Moscibroda, T., Wattenhofer, R.: Maximal independent sets in radio networks.
\newblock In: Proceedings of the 24th Annual ACM Symposium on Principles of
  Distributed Computing, PODC '05, pp. 148--157. ACM, New York, NY, USA (2005)

\bibitem{Ni2011c}
Ni, J., Srikant, R., Wu, X.: Coloring spatial point processes with applications
  to peer discovery in large wireless networks.
\newblock IEEE/ACM Transactions on Networking \textbf{19}(2), 575 --588 (2011)

\bibitem{Panconesi2001s}
Panconesi, A., Rizzi, R.: Some simple distributed algorithms for sparse
  networks.
\newblock Distributed Computing \textbf{14}(2), 97--100 (2001)

\bibitem{Panconesi1996o}
Panconesi, A., Srinivasan, A.: On the complexity of distributed network
  decomposition.
\newblock Journal of Algorithms \textbf{20}(2), 356--374 (1996)

\bibitem{Park1996a}
Park, T., Lee, C.Y.: Application of the graph coloring algorithm to the
  frequency assignment problem.
\newblock Journal of the Operations Research Society of Japan-Keiei Kagaku
  \textbf{39}(2), 258--265 (1996)

\bibitem{Peleg2000d}
Peleg, D.: Distributed computing: a locality-sensitive approach.
\newblock Society for Industrial and Applied Mathematics, Philadelphia, PA, USA
  (2000)

\bibitem{Prakash1997a}
Prakash, R., Raynal, M., Singhal, M.: An adaptive causal ordering algorithm
  suited to mobile computing environments.
\newblock Journal of Parallel and Distributed Computing \textbf{41}(2),
  190--204 (1997)

\bibitem{Schneider2008a}
Schneider, J., Wattenhofer, R.: A log-star distributed maximal independent set
  algorithm for growth-bounded graphs.
\newblock In: Proceedings of the 27th ACM Symposium on Principles of
  Distributed Computing, PODC '08, pp. 35--44. ACM, New York, NY, USA (2008)

\bibitem{Schneider2010a}
Schneider, J., Wattenhofer, R.: A new technique for distributed symmetry
  breaking.
\newblock In: Proceedings of the 29th ACM SIGACT-SIGOPS Symposium on Principles
  of Distributed Computing, PODC '10, pp. 257--266. ACM, New York, NY, USA
  (2010)

\bibitem{Scott2013f}
Scott, A., Jeavons, P., Xu, L.: Feedback from nature: an optimal distributed
  algorithm for maximal independent set selection.
\newblock In: Proceedings of the 2013 ACM Symposium on Principles of
  Distributed Computing, PODC '13, pp. 147--156. ACM, New York, NY, USA (2013)

\bibitem{Waters2005g}
Waters, R.J.: Graph colouring and frequency assignment.
\newblock Ph.D. thesis (2005)

\bibitem{Wattenhofer2007l}
Wattenhofer, R.: \url{http://dcg.ethz.ch/lectures/fs08/distcomp/lecture/
  chapter4.pdf} (2007)

\bibitem{Xu2013s}
Xu, L., Jeavons, P.: Simple neural-like {P} systems for maximal independent set
  selection.
\newblock Neural Computation \textbf{25}(6), 1642--1659 (2013)

\bibitem{Zuckerman2006l}
Zuckerman, D.: Linear degree extractors and the inapproximability of max clique
  and chromatic number.
\newblock In: Proceedings of the 38th Annual ACM Symposium on Theory of
  Computing, STOC '06, pp. 681--690. ACM, New York, NY, USA (2006)

\end{thebibliography}
\end{document}